\documentclass[letterpaper,11pt,reqno]{amsart}

\usepackage{amsmath}%
\usepackage{amssymb}%

\usepackage{diagmac2}

\newtheorem{theorem}{Theorem}[section] 

\newtheorem{prop}[theorem]{Proposition}
\newtheorem{cor}[theorem]{Corollary}
\newtheorem{lemma}[theorem]{Lemma}

\theoremstyle{plain}

\newtheorem{example}{Example}

\theoremstyle{definition}
\newtheorem{definition}{Definition}[section] 

\theoremstyle{remark}
\newtheorem{remark}{Remark}[section]

\numberwithin{equation}{section}

\newcommand{\impl}{\rightarrow}
\newcommand{\eqv}{\leftrightarrow}

\newcommand{\Z}{\mathsf{Z}}

\newcommand{\algA}{\mathsf{A}}
\newcommand{\algB}{\mathsf{B}}
\newcommand{\algC}{\mathsf{C}}

\newcommand{\alga}{\mathsf{a}}
\newcommand{\algb}{\mathsf{b}}
\newcommand{\algc}{\mathsf{c}}
\newcommand{\algd}{\mathsf{d}}
\newcommand{\algf}{\mathsf{f}}
\newcommand{\algt}{\mathsf{t}}
\newcommand{\algg}{\mathsf{g}}
\newcommand{\algr}{\mathsf{r}}

\newcommand{\oalga}{\overline{\mathsf{a}}}
\newcommand{\oalgb}{\overline{\mathsf{b}}}

\newcommand{\oalgd}{\overline{\mathsf{d}}}

\newcommand{\ualga}{\underline{\mathsf{a}}}

\newcommand{\classK}{\mathcal{K}}
\newcommand{\classV}{\mathcal{V}}

\newcommand{\Heyt}{\mathcal{H}}

\newcommand{\Topo}{\mathcal{I}}

\newcommand{\classI}{\mathcal{I}}
\newcommand{\classM}{\mathcal{M}}

\newcommand{\CHom}{\boldsymbol{\mathsf{H}}}
\newcommand{\CSub}{\boldsymbol{\mathsf{S}}}

\newcommand{\one}{\mathbf{1}}
\newcommand{\zero}{\mathbf{0}}

\newcommand{\lbr}{\langle}
\newcommand{\rbr}{\rangle}

\newcommand{\bydef}{\leftrightharpoons}    
\newcommand{\dimpl }{\Rightarrow}

\newcommand{\op}{\overline{p}}
\newcommand{\up}{\underline{p}}
\newcommand{\uq}{\underline{q}}

\newcommand{\ua}{\underline{\alga}}

\begin{document}

\renewcommand{\footnoterule}{\noindent\rule{5cm}{0.4pt}{\vspace{5pt}}}%
\renewcommand\thefootnote{\arabic{footnote}}%
\setcounter{footnote}{0}%

\title[Characteristic formulas]{Characteristic formulas over intermediate logics}

\author{Alex Citkin}
\email{acitkin@gmail.com}

\begin{abstract}
We expand the notion of characteristic formula to infinite finitely presentable subdirectly irreducible algebras. We prove that there is a continuum of varieties of Heyting algebras containing infinite finitely presentable subdirectly irreducible algebras. Moreover, we prove that there is a continuum of intermediate logics that can be axiomatized by characteristic formulas of infinite algebras while they are not axiomatizable by standard Jankov formulas. We give the examples of intermediate logics that are not axiomatizable by characteristic formulas of infinite algebras. Also, using the G\"odel-McKinsey-Tarski translation we extend these results to the varieties of interior algebras and normal extensions of S4.  
\end{abstract}


\maketitle

\section{Introduction}

One of the very useful notions in research of intermediate logics and Heyting algebras is a notion of characteristic or Jankov formula introduced in \cite{Jankov_1963_ded}. With every finite subdirectly irreducible (s.i.) algebra $\algA$, using a diagram of algebra $\algA$, one can construct a formula $\chi(\algA)$ that enjoys the following properties:

\begin{center}
\begin{tabular}{l p{10cm}}
(Hom) & if formula $\chi(\algA)$ is refutable in algebra $\algB$ (in symbols $\algB \nvDash A$), then algebra $\algA$ is embeddable into some quotient algebra of algebra $\algB$ (i.e. $\algA$ is homo-embeddable in $\algB$);\\
(Ded) & if formula $A$ is refutable in algebra $\algA$, then $A \vdash \chi(\algA)$ in intuitionistic propositional calculus IPC, i.e. characteristic formula is the weakest relative to derivability formula among formulas refutable in $\algA$;\\
(Irr)  & characteristic formula $\chi(\algA)$ is meet-prime, that is for any formulas $A,B$ if $A \land B \vdash \chi(\algA)$ then $A \vdash \chi(\algA)$ or $b \vdash \chi(\algA)$.
\end{tabular}
\end{center}

Independently, and about at the same time, formulas with similar properties but constructed based upon finite frames rather than algebras, were introduced by K. Fine for modal logics \cite{Fine_Asc_1974} and D. de Jongh for intermedite logics \cite{deJongh_Th}. Then the theory of frame based formulas was extended further to different types of subframes (for more details see e.g. \cite{Chagrov_Zakh,Bezhanishvili_N_PhD,Yang_PhD}).

Let us observe that (Hom) entails that every Jankov or frame formula defines a splitting\footnote{For definition see \cite{McKenzie_1972}.} variety. Thus, every Jankov or frame formula is meet-prime, that is, \(A,B \vdash C\) yields \(A \vdash C\) or \(B \vdash C\) for any formulas \(A,B,C\). So, it was natural to try and use Jankov or frame formulas as the building blocks for axiomatization of intermediate or modal logics. It turned out though that not every intermediate logic can be axiomatized by Jankov formulas\footnote{An example can be found in \cite[Proposition 9.5]{Chagrov_Zakh}. In fact, there is a continuum of intermediate logics that cannot be axiomatized by Jankov formulas \cite[Corollary from the Theorem 4.8]{Tomaszewski_PhD}.}. In \cite{Zakh_intermediate_1983,Zakh_Sintax_1988,Zakharyashchev_Sintax_1989} M.Zakharyaschev introduced canonical formulas, using which one can axiomatize any intermediate logic or any logic of transitive frames (see also  \cite{Zakh_Canonical_1992,Zakh_Canonical_1996,Zakh_Canonical_1997}). Canonical formulas proved to be very helpful in solving some problems in intermediate or modal logics (see e.g. \cite{Chagrov_Zakh}). An algebraic account of the theory of canonical formulas was offered in \cite{Tomaszewski_PhD,Bezhanishvili_G_N_2009} where canonical formulas are regarded as modified Jankov formulas of finite subdirectly irreducible algebras\footnote{The difference between these two approaches is outlined in \cite{Bezhanishvili_G_N_2009}.}. 

The notion of characteristic formula can be extended to finitely presentable s.i. algebras\footnote{The idea of using a presentation of algebra instead of its diagram was introduced in \cite{Citkin1977} for quasi-characteristic rules, and for varieties with equationally definable principal congruences (EDPC) a finite presentation was used in \cite{Blk_Pgz_1}, see also \cite{Wolter_PhD}[Definition 2.4.10].}. Utilizing a diagram formula is just one particular way to constructing some presentation. Using different formulas  defining algebra $\algA$ (as defining relation) one can construct syntactically different formulas each of which is interderivable with $\chi(\algA)$ and many of which are syntactically simpler than $\chi(A)$.  

Naturally a question arises whether one can use finite presentation for infinite algebras and in such a way to expand the notion of characteristic formula. The obvious negative answer to this question follows from \cite{Blk_Pgz_1} where it was observed that in finitely approximated varieties with EDPC (and, as it is well known, the variety $\Heyt$ of all Heyting algebras is finitely approximated and enjoy EDPC) every s.i. finitely presentable algebra is finite. Let us note the very important property of finite presentability: finite presentability is relative to a given variety and an algebra can be finitely presentable over some varieties while being not finitely presentable over others\footnote{The Heyting algebras finitely presentable over variety of all Heyting algebras are studied in \cite{Butz_1998}.}. 

Let us also recall that not every variety of Heyting algebras is finitely approximated. In fact, in \cite{Jankov_1968} using characteristic formulas V.A. Jankov proved that there is a continuum of not finitely approximated intermediate logics, thus, there is a continuum of not finitely approximated varieties of Heyting algebras. Therefore one can ask whether there are varieties of Heyting algebras containing infinite finitely presentable over them s.i. algebras. We give the positive answer to this question and we demonstrate that there is a continuum varieties of Heyting algebras containing infinite finitely presentable s.i. algebras. So one can construct the characteristic formulas relative to a particular variety and such relative characteristic formulas are not Jankov formulas. We also prove that there is a continuum of varieties of Heyting algebras (or intermediate logics) that can be axiomatized by relative characteristic formulas, but cannot be axiomatized by Jankov formulas.

M.Kracht \cite{Kracht_1990} was using finitely presentable algebras as a main tool while studying the splittings in non-transitive modal logics. As he pointed out the situation with non-transitive logics is totally different from transitive case since in non-transitive logics ``it is no longer true that only the finitely presentable, subdirectly irreducible (s.i.) algebras induce splittings''. Later it was observed that there are not finitely presentable splitting algebras in extensions of \textbf{GL} (cf. \cite{Kracht_Tools_1999}[Theorem 7.5.16]) and even in extensions of \textbf{S4} or \textbf{IPC} (cf. \cite{Citkin_Not_2012}). In the present paper we are not concerned with splitting algebras. Our goal is to demonstrate that in some intermediate logics, or in some varieties of Heyting algebras for this matter, there are infinite subdirectly irreducible finitely presentable algebras and we can make a use of characteristic formulas associated with this algebras. 

\section{Basic definitions}

We will consider Heyting algebras in the signature $\left\{\land, \lor, \impl, \neg \right\}$ and use \(\eqv\) as abbreviation: \(A \eqv B \bydef (A \impl B) \land (B \impl A)\). By $\Z_n$ we denote a $n$-element 1-generated Heyting algebra, so $\Z_2$ is a two-element Boolean algebra, $\Z_\infty$ is a Rieger-Nishimura ladder that we also will denote by $\Z$. If $\alga \in \algA$ by $\nabla(\alga)$ we denote a filter generated by element $\alga$.  

If $\algA, \algB$ are algebras by $\algA + \algB$ we denote a\textit{ concatenation}\footnote{The concatenations are often called ordered, linear or Troelstra sums. We are trying to avoid use of the term ``sum'' since it suggests some kind of commutativity which is not the case here.} of $\algA$ and $\algB$, that is $\algA + \algB$ is an algebra obtained by putting algebra $\algB$ onto $\algA$ and ``gluing'' the top element of $\algA$ and the bottom element of $\algB$. Let us observe the following rather simple property of concatenation.

\begin{prop} \label{concatf}If $\algA$ and $\algB$ are algebras and $\nabla \subseteq \algB$ is a filter, then
\begin{equation}
(\algA + \algB)/\nabla \cong \algA + \algB/\nabla.
\end{equation}
\end{prop}

Since in this paper we consider only Heyting algebras, we will simply say ``algebra''. By $\Heyt$ we denote a variety of all Heyting algebras.  

If $\algA$ is an algebra then $\classV(\algA)$ is a variety generated by algebra $\algA$. If $A$ and $B$ are (propositional) formulas and $p$ is a variable then by $A(B/p)$ we denote a result of the substitution of formula $B$ for variable $p$ in formula $A$. Strings of distinct variables are indicated by $\up, \uq$ and if $A$ contains variables only from the list $\up = p_1,\dots,p_n$, we express this fact by the notation $A(\up)$ or $A(p_1,\dots,p_n)$. Accordingly, if $\alga_1,\dots,\alga_n$ are elements of some algebra and $A(\up)$ is a formula we can write $A(\ualga)$ instead of $A(\alga_1,\dots,\alga_n)$.

Element $\alga$ is said to be \textit{regular} \cite{Rasiowa_Sikorski} if $\neg \neg \alga = \alga$, and element $\alga$ is said to be \textit{dense} \cite{Rasiowa_Sikorski} if $\neg \neg \alga = \one$. The set of all regular elements of algebra $\algA$ we denote by $Rg(\algA)$, the set of all dense elements of algebra $\algA$ we denote by $Dn(\algA)$. Clearly $Dn(\algA)$ is such a filter of algebra $\algA$ that $\algA/Dn(\algA)$ is a Boolean algebra and the natural homomorphism sends all dense elements of $\algA$ in $\one$. Moreover, $Dn(\algA)$ isomorphic as a lattice to $R(\algA)$ (e.g.\cite{Rasiowa_Sikorski}).

If $\classV$ is a variety and $A$ is a formula then by $\vdash_\classV A$ we denote the fact that the formula $A$ is valid in $\classV$, that is $\algA \vDash A$ for all $\algA \in \classV$. If $\classK$ is a class of algebras, by $\CSub\classK$ we denote a class of all subalgebras of algebras from $\classK$ and by $\CHom\classK$ we denote a class of all homomorphic images of all algebras from $\classK$. If class $\classK$ consists of just one algebra $\algA$ we will write $\CSub\algA$ and $\CHom\algA$.

Let us recall the following definition.

\begin{definition}\label{fpdef} (cf. \cite{MaltsevBook}) Let $\classV$ be a variety, $A(\up)$ be a formula and $\nu$ be a valuation in algebra $\algA$. Then a pair $\lbr A, \nu \rbr$ \textit{defines algebra} $\algA$ over $\classV$ if
\begin{enumerate}
\item Elements $\alga_1 = \nu(p_1),\dots, \alga_n = \nu(p_n)$ generate algebra $\algA$;
\item $A(\ua) = \one$;
\item For any formula $B(\up)$ if $B(\ua) = \one$ then $\vdash_\classV(A(\up) \impl B(\up))$.
\end{enumerate} 
Algebra $\algA$ is said to be \textit{finitely presentable over variety} $\classV$ if there exists a pair that defines algebra $\algA$ over variety $\classV$. We also will say that formula $A(\op)$ defines algebra $\algA$ over variety $\classV$ in generators $\ualga$ or that a pair $\lbr A,\nu \rbr$ is a presentation of algebra $\algA$ over $\classV$ (sometimes we will omit reference to $\classV$). 
\end{definition}

The following criterion is very useful.

\begin{prop}\cite{MaltsevBook}. \label{homext} Let $\classV$ be a variety, $A(\up)$ be a formula and $\nu$ be a valuation in algebra $\algA$. Then a pair $\lbr A, \nu \rbr$ \textit{defines algebra} $\algA$ over $\classV$ if and only if
\begin{enumerate}
\item Elements $\alga_1 = \nu(p_1),\dots, \alga_n = \nu(p_n)$ generate algebra $\algA$;
\item $A(\ua) = \one$;
\item if $\algB \in \classV$ is an algebra, $\algb_1,\dots,\algb_n \in \algB$ and $A(\algb_1,\dots,\algb) = \one_\algB$ then the mapping $\alga_i \mapsto \algb_i; i=1,\dots,n$ can be extended to homomorphism of $\algA$ in $\algB$. 
\end{enumerate} 
\end{prop}

\begin{remark} Since any variety is closed relative to homomorphisms, in Proposition \ref{homext} it is sufficient take into consideration only s.i. algebras $\algB \in \classV$. 
\end{remark}

If $\classV$ is a variety, by $SI(\classV)$ we will denote a class of all s.i. algebras from $\classV$, by $FP(\classV)$ - a class of all algebras finitely presentable over $\classV$, and by $FPSI(\classV)$ - a class of all finitely presentable over $\classV$ s.i. algebras.

Let us note also the following properties of finitely presentable algebras.

\begin{prop} \label{iso_1} If pairs $\lbr A(\up), \nu \rbr$ and $\lbr B(\up), \nu \rbr$ define over $\classV$ the same algebra $\algA$ then $\vdash_\classV A \eqv B$.
\end{prop}
\begin{proof} Straight from the Definition \ref{fpdef}(3). \end{proof}

\begin{prop} \cite[Theorem 5, Chap.V sect.11]{MaltsevBook} \label{dyck} Assume that formulas $A(\up)$ and $B(\up)$ define over variety $\classV$ algebras $\algA$ and $\algB$. Then if $\vdash_\classV B \impl A$ then $\algB$ is a homomorphic image of $\algA$. In particular, if $\vdash_\classV A \eqv B$ the algebras $\algA$ and $\algB$ are isomorphic. 
\end{prop}

\begin{prop} \cite[Corollary 7, Chap.V sect.11]{MaltsevBook} \label{tietze} If an algebra $\algA$ is finitely presentable over $\classV$ then $\algA$ is finitely presentable in any set of its generators. 
\end{prop}


\section{Characteristic formulas}

Let us recall the definition of Jankov formula.

\begin{definition}\cite{Jankov_1968} Assume $\algA$ is a finite s.i. algebra and $\algA = \left\{\alga_1,\dots,\alga_n \right\}$. With every element $\alga_i \in \algA$ we associate a variable $p_i; \ i=1,\dots,n$. Let
\begin{equation}
\begin{split}
D(p_1,\dots,p_n) = & \bigwedge_{\alga_i \land \alga_j = \alga_k} (p_i \land p_j \eqv p_k) \land \\
 & \bigwedge_{\alga_i \lor \alga_j = \alga_k} (p_i \lor p_j \eqv p_k) \land \\
 & \bigwedge_{\alga_i \impl \alga_j = \alga_k} (p_i \impl p_j \eqv p_k) \land \\
 & \bigwedge_{\neg \alga_i = \alga_j} (\neg p_i \eqv p_j).
\end{split} \label{Jankov}
\end{equation}
for all $i,j,k \in \left\{1,\dots,n \right\}$. Formulas $D$ is a diagram formula of algebra $\algA$. Let $\alga_n$ be a \textit{opremum} of algebra $\algA$, that is, the greatest element among all distinct from $\one$ elements of $\algA$. Then formula
\begin{equation}
\chi(\algA) = D(p_1,\dots,p_n) \impl p_n
\end{equation}   
is called \textit{Jankov formula}.
\end{definition}

One of the most frequently used properties of Jankov formulas are presented in the following Proposition.

\begin{prop} \cite{Jankov_1969} \label{char} Assume $\algA$ is a finite s.i. algebra, $\algB$ is an algebra and $B$ is a formula. Then
\begin{itemize}
\item[(a)] if $\algB \nvDash \chi(\algA)$, then $\algA \in \CSub\CHom\algB$;
\item[(b)] if $\algA \nvDash B$, then $B \vdash_{IPC} \chi(\algA)$.
\end{itemize} 
\end{prop}

The property (b) from the Proposition \ref{char} means that $\chi(\algA)$ is a weakest formula among formulas refutable in \(\algA\). Let us note that if $A_1$ and $A_2$ are two weakest formulas refutable in \(\algA\), then formulas $A_1$ and $A_2$ are inter-derivable in IPC. 

Now let us extend the definition of Jankov formulas to finitely presentable algebras.

\begin{definition} Assume $\classV$ is a variety and $\algA \in \classV$ is a s.i. algebra finitely presentable over $\classV$. Suppose $\lbr A(\up), \nu\rbr$ is a presentation of $\algA$ over $\classV$. If $B(\up)$ is such a formula that $\nu(B)$ is an opremum of $\algA$, then the formula
\begin{equation}
\chi_{_\classV}(\algA) = A(\up) \impl B(\up)
\end{equation}
is a \textit{characteristic formula of algebra }$\algA$ \textit{over variety }$\classV$.
\end{definition} 

First of all, let us note that since $\algA$ is s.i., it always has an opremum. And, since elements $\nu(p_1),\dots,\nu(p_n)$ generate algebra $\algA$ there alway is such a formula $B(\up)$ that $\nu(B)$ is an opremum. Thus, for any s.i. finitely presentable over $\classV$ algebra one can define a characteristic formula. Let us establish properties of characteristic formulas similar to those of Jankov formulas.

\begin{theorem} \label{charf} Assume $\algA \in FPSI(\classV)$, $\algB \in \classV$ is an algebra and $B$ is a formula. Then
\begin{itemize}
\item[(a)] if $\algB \nvDash \chi_{_\classV}(A)$, then $\algA \in \CSub\CHom\algB$;
\item[(b)] if $\algA \nvDash B$, then $B \vdash_\classV \chi_{_\classV}(A)$.
\end{itemize} 
\end{theorem}
\begin{proof} 
(a) Suppose $\algB \nvDash \chi_{_\classV}(A)$. By definition of characteristic formula $\chi_{_\classV}(\algA) = A(\up) \impl B(\up)$ where $\lbr A(\up), \nu \rbr$ is a defining pair and $\nu(B(\up)) = \alga \in \algA$ is an opremum of algebra $\algA$. Thus, since $\algB \nvDash (A(\up) \impl B(\up))$, for some homomorphic image $\algB'$ of algebra $\algB$ and some elements $\algb_1,\dots,\algb_n \in \algB'$ we have
\begin{equation}
A(\algb_1,\dots,\algb_n) = \one_\algB \text{ and } B(\algb_1,\dots,\algb_n) \neq \one_\algB. \label{refut1}
\end{equation}
By Proposition \ref{homext} the mapping $\phi: \nu(p_i) \mapsto \algb_i; \ i=1,\dots,n$ can be extended to homomorphism $\phi: \algB \impl \algB'$. Let us observe that 
\begin{equation}
\phi(\nu(B)) = B(\algb_1,\dots,\algb_n) \neq \one_\algB. \label{pretop}
\end{equation}
Recall that the opremum of a Heyting algebra is in a kernel of any proper homomorphism. Hence, from \eqref{pretop} it follows that $\phi$ is a isomorphism. Thus, algebra $\algA$ is embeddable into $\algB'$ and $\algA \in \CSub\CHom\algB$.   

(b) Assume the contrary: $\algA \nvDash B$ and  $B \nvdash_\classV \chi_{_\classV}(A)$. If $B \nvdash_\classV \chi_{_\classV}(A)$, then for some algebra $\algB \in \classV$ we have 
\begin{equation}
\algB \vDash B \text{ and } \algB \nvDash \chi_{_\classV}(\algA). 
\end{equation}
But we just have proven that, if $\algB \nvDash \chi_{_\classV}(\algA)$, then $\algA \in \CSub\CHom\algB$. Thus, if $\algB \vDash B$ then $\algA \vDash B$ and this contradict the assumption.   
\end{proof}
\begin{cor} \label{unique} Let $\classV$ be a variety and $\algA \in FPSI(\classV)$. Then all characteristic formulas of algebra $\algA$ over $\classV$ (regardless of which presentation we used) are inter-derivable over $\classV$.
\end{cor}

The Corollary \ref{unique} means that a finitely presentable over $\classV$ s.i. algebra defines unique modulo inter-derivability in \(\classV\) characteristic formula.
 

Let $\classV$ be a variety. A set of formulas is called \textit{$\classV$-independent} if no one formula of this set is derivable over $\classV$ from the rest of formulas. A $\Heyt$-independent set of formulas we will call \textit{independent}. 

On the set $\Heyt$ of all Heyting algebras we can define the following quasi-order: $\algA \leq \algB$ if $\algA \in \CSub\CHom\algB$. The reflexivity of $\leq$ is trivial, while transitivity follows from the fact that variety of Heyting algebras has a congruence extension property (see, for instance,  \cite{Blk_Pgz_1}). A class $\classK$ of algebras is said to be an \textit{antichain} if for any $\algA,\algB \in \classK$ we have $\algA \not\leq \algB$ and $\algB \not\leq \algA$   

\begin{cor}\label{indep} Let $\classV$ be a variety and $\classK \subseteq FPSI(\classV)$. $\classK$ is an antichain  if and only if the set $\left\{\chi_{_\classV}(\algA); \ \algA \in \classK\right\}$ is $\classV$-independent.
\end{cor}
\begin{proof} Let $\classK \subseteq FPSI(\classV)$ be an antichain. Then if $\algA \in \classK$ we have $\algA \nvDash \chi_{_\classV}(\algA)$, but $\algB \vDash \chi_{_\classV}(\algA)$, because $\classK$ is an antichain and, by Theorem \ref{charf}(a) $\algA \notin \CSub\CHom\algB$. 

Conversely, assume the contrary: $\algB \in \CSub\CHom\algA$. Then, since $\algB \nvDash \chi_{_\classV}(\algB)$, we have $\algA \nvDash \chi_{_\classV}(\algB)$. By virtue of Theorem \ref{charf}(b), $\chi_{_\classV}(\algB) \vdash_\classV \chi_{_\classV}(\algA)$. And the latter contradicts $\classV$-independence. 
\end{proof}

Let us note that if \(\classV_1,\classV_2\) are varieties and \(\classV_1 \subseteq \classV_2\) then \(\classV_1\)-independence yields \(\classV_2\)-independence. Thus, for every variety \(\classV\) any $\classV$-independent set of formulas is independent.

The following corollary is a consequence of the previous one.

\begin{cor}\label{indep1} Let $\classV$ be a variety and $\classK \subseteq FPSI(\classV)$ and $\classK$ is an antichain. Then set $\left\{\chi_{_\classV}(\algA); \ \algA \in \classK\right\}$ is independent.
\end{cor}

\begin{remark} It is obvious that any finite s.i. algebra $\algA$ is finitely presentable (over $\Heyt$). Let us observe that diagram formula $D$ in the definition of Jankov formula \eqref{Jankov} defines algebra $\algA$ in the trivial set of generators: the set of all elements of algebra $\algA$. On the other hand, if as a set of generators we take a set of all distinct from $\one$ $\lor$-irreducible elements and use a diagram relations in order to construct a defining formula, we obtain a formula interderivable with de Jongh formula \cite{deJongh_Th} (or frame-based formula \cite{Bezhanishvili_N_PhD}) of algebra $\algA$.    
\end{remark}

Let us recall \cite{McKenzie_1972} that an algebra \(\algA\) from a variety \(\classV\) is call \textit{splitting in the variety} \(\classV\) if there is the greatest subvariety of \(\classV\) not containing algebra \(\algA\).

\begin{remark} \label{split} From Theorem \ref{charf} it is immediately follows any algebra from \(FPSI(\classV)\) is a splitting in the variety \(\classV\).
\end{remark}


\section{ An example of infinite finitely presentable s.i. Heyting algebra}

The goal of this section is to give an example of a variety $\classV$ and an infinite algebra $\Z' \in FPSI(\classV)$. Then, in the following section, based on this example we will construct an infinite set of such algebras. More precisely, we will show that algebra $\Z' = \Z \times \Z_2 + \Z_2$ depicted in Fig. 1(a) (the corresponding frame is depicted in Fig. 1(c)) is finitely presentable over every variety generated by an algebra $\Z \times \Z_2 + \algA$, where $\algA$ is any non-degenerate algebra.

\[
\ctdiagram{
\ctnohead
\ctinnermid
\ctel 0,15,15,0:{}
\ctel 15,0,45,30:{}
\ctel 0,15,45,60:{}
\ctel 0,15,15,30:{}
\ctel 30,15,0,45:{}
\ctel 45,30,5,70:{}
\ctel 0,45,25,70:{}
\ctel 45,60,35,70:{}
\ctel 0,25,15,10:{}
\ctel 15,10,45,40:{}
\ctel 0,25,45,70:{}
\ctel 0,25,15,40:{}
\ctel 30,25,0,55:{}
\ctel 45,40,10,75:{}
\ctel 0,55,20,75:{}
\ctel 45,70,40,75:{}
\ctel 20,110,20,125:{}
\ctel 15,95,20,100:{}
\ctel 15,105,20,110:{}
\ctel 25,95,20,100:{}
\ctel 25,105,20,110:{}
\ctel 15,0,15,10:{}
\ctel 0,15,0,25:{}
\ctel 30,15,30,25:{}
\ctel 45,30,45,40:{}
\ctel 15,30,15,40:{}
\ctel 30,45,30,55:{}
\ctel 45,60,45,70:{}
\ctel 0,45,0,55:{}
\ctel 15,60,15,70:{}
\ctel 20,100,20,110:{}
\ctv 15,0:{\bullet}
\ctv 0,15:{\bullet}
\ctv 30,15:{\bullet}
\ctv 35,15:{\oalga}
\ctv 45,30:{\bullet}
\ctv 15,30:{\bullet}
\ctv 30,45:{\bullet}
\ctv 45,60:{\bullet}
\ctv 0,45:{\bullet}
\ctv 15,60:{\bullet}
\ctv 20,85:{. \ . \ . \ . \ . \ . \ . \ . \ .}
\ctv 20,100:{\bullet}
\ctv 15,10:{\bullet}
\ctv 15,17:{\oalgb}
\ctv 0,25:{\bullet}
\ctv 30,25:{\bullet}
\ctv 45,40:{\bullet}
\ctv 15,40:{\bullet}
\ctv 30,55:{\bullet}
\ctv 45,70:{\bullet}
\ctv 55,70:{\oalga^7}
\ctv 0,55:{\bullet}
\ctv 15,70:{\bullet}
\ctv 20,90:{\overbrace{. \ . \ . \ . \ . \ . \ . \ . \ .}^{}}
\ctv 20,110:{\bullet}
\ctv 20,125:{\bullet}
\ctel 100,15,115,0:{}
\ctel 115,0,145,30:{}
\ctel 115,0,145,30:{}
\ctel 100,15,145,60:{}
\ctel 100,15,115,30:{}
\ctel 130,15,100,45:{}
\ctel 145,30,110,65:{}
\ctel 100,45,120,65:{}
\ctel 145,60,140,65:{}
\ctv 115,0:{\bullet}
\ctv 115,-10:{0}
\ctv 100,15:{\bullet}
\ctv 93,15:{\algr_2}
\ctv 130,15:{\bullet}
\ctv 137,15:{\algg}
\ctv 145,30:{\bullet}
\ctv 153,30:{\algr_1}
\ctv 115,30:{\bullet}
\ctv 105,30:{\algt}
\ctv 130,45:{\bullet}
\ctv 145,60:{\bullet}
\ctv 100,45:{\bullet}
\ctv 115,60:{\bullet}
\ctv 120,80:{\overbrace{. \ . \ . \ . \ . \ . \ . \ . \ .}^{}}
\ctv 120,90:{\bullet}
\ctv 130,90:{1}
\ctel 225,0,240,15:{}
\ctel 225,0,210,15:{}
\ctel 195,45,195,95:{}
\ctel 225,45,225,95:{}
\ctel 195,95,225,65:{}
\ctel 195,80,225,50:{}
\ctel 195,65,210,50:{}
\ctel 195,50,200,45:{}
\ctel 225,95,195,80:{}
\ctel 225,80,195,65:{}
\ctel 225,65,195,50:{}
\ctel 225,50,220,45:{}
\ctv 225,0:{\bullet}
\ctv 210,15:{\bullet}
\ctv 240,15:{\bullet}
\ctv 247,15:{\oalgb}
\ctv 195,35:{.}
\ctv 195,40:{.}
\ctv 195,45:{.}
\ctv 195,50:{\bullet}
\ctv 195,65:{\bullet}
\ctv 195,80:{\bullet}
\ctv 195,95:{\bullet}
\ctv 225,35:{.}
\ctv 225,40:{.}
\ctv 225,45:{.}
\ctv 225,50:{\bullet}
\ctv 225,65:{\bullet}
\ctv 225,80:{\bullet}
\ctv 225,95:{\bullet}
\ctv 233,95:{\oalga}
\ctv 15,-25:{(a)}
\ctv 115,-25:{(b)}
\ctv 210,-25:{(c)}
}
\]

\begin{center}
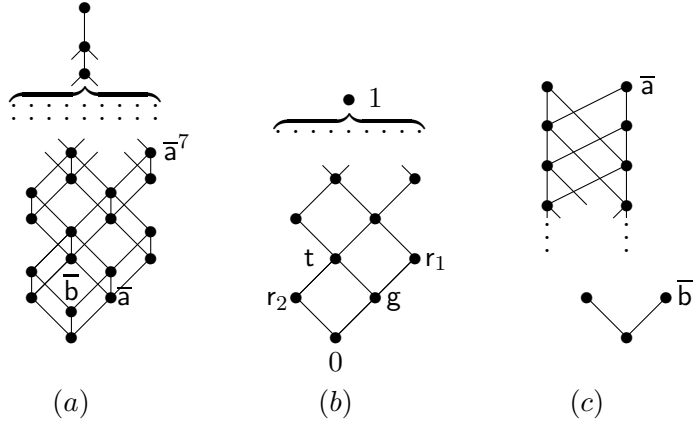
Figure 1. An example of s.i. finitely presentable algebra
\end{center}

The elements of algebra $\Z \times \Z_2$ we will regard as pairs $\lbr\alga,\algb\rbr$ where $\alga \in \Z$ and $\algb \in \left\{0,1\right\}$. So, $\oalgb = \lbr0,1 \rbr$ and $\lbr 1,1 \rbr$ is an opremum of algebra $\Z'$ (see Fig.1). Let us note that elements $\oalga, \oalgb$ generate algebra $\Z'$.

Our goal is to prove the following theorem.

\begin{theorem} \label{mainth} Let $\algA$ be any non-degenerate Heyting algebra and $\Z^* = \Z \times \Z_2 + \algA$. Then algebra $\Z'$ is finitely presentable over variety $\classV(\Z^*)$.
\end{theorem}

We will prove that the formula
\begin{equation}
A = \neg(p \land q) \land (\neg \neg q \impl q) \land ( p^{10} \impl (q \lor \neg q)),
\end{equation}
where $p^{10} = (\neg \neg p \impl p) \lor ((\neg \neg p \impl p) \impl (p \lor \neg p))$, and the valuation $\nu : p \mapsto \oalga, \ \nu : q \mapsto \oalgb$ define algebra $\Z'$ over $\classV$.

Let us observe that formula $A$ is equivalent to the following formula:
\begin{equation}
\begin{split}
& \neg(p \land q) \land (\neg \neg q \impl q) \land \\
& ((\neg \neg p \impl p )\impl (q \lor \neg q)) 
\land (((\neg \neg p \impl p) \impl (p \lor \neg p)) \impl (q \lor \neg q)). \label{presform}
\end{split}
\end{equation}

In order to prove the theorem we need to establish some properties of formulas valid on elements $\oalga, \oalgb \in \Z'$.
\subsection{Auxiliary lemmas}

\begin{lemma} \label{lcase_0} Suppose $B(p,q)$ is a formula and $B(\oalga,\oalgb) = \one_{\Z'}$. Then 
\begin{equation}
IPC \vdash B(p,q \land \neg q). \label{case_0}
\end{equation}
\end{lemma}
\begin{proof} Recall that \(\oalga = \lbr \algg, \zero \rbr\) and \(\oalgb = \lbr \zero, \one \rbr\). Thus from $B(\oalga,\oalgb) = \one$ we have $B(\lbr \algg, \zero \rbr,\lbr \zero, \one \rbr) = \one$, hence $B(\algg,\zero) = \one_{\Z}$ and, obviously, $B(\algg, \algg \land \neg \algg) = \one_\Z$. Let us also recall that $\Z$ is a free algebra and $\algg$ is its free generator (see Fig.1), hence, $IPC \vdash B(p, p \land \neg p)$. Taking into consideration that  $IPC \vdash (p \land \neg p) \eqv (q \land \neg q)$, we can conclude that $IPC \vdash B(p, q \land \neg q)$.  
\end{proof}


\begin{lemma}  Suppose $B(p,q)$ is a formula and $B(\oalga,\oalgb) = \one_{\Z}'$. Then in the 2-element Boolean algebra \(\Z_2\) 
\begin{equation*}
B(\zero,\zero) = B(\zero,\one) = B(\one,\zero) =\one.
\end{equation*}
\end{lemma}
\begin{proof} Let us consider the following three filters : $\nabla(\neg \oalga \land \neg \oalgb),\nabla(\oalgb), \nabla(\oalga)$. And now let us observe that corresponding homomorphisms send elements $\oalga,\oalgb$ respectively in $\zero_\algB,\zero_\algB$, or $\zero_\algB,\one_\algB$, or $\one_\algB,\zero_\algB$. Since $B(\oalga,\oalgb) = \one_{\Z'}$ and any homomorphism preserves the top element, we can complete the proof.   
\end{proof}

\begin{cor} \label{BoolB} If $B(p,q)$ is a formula and $B(\oalga,\oalgb) = \one_{\Z}'$ then in any Heyting algebra
\begin{equation*}
B(\zero,\zero) = B(\zero,\one) = B(\one,\zero) =\one.
\end{equation*}
\end{cor}
\begin{proof} Recall that in any Heyting algebra the set $\left\{\zero,\one \right\}$ forms a subalgebra isomorphic with \(\Z_2\) .
\end{proof}

\begin{lemma} \label{lcase_1} Suppose $B(p,q)$ is a formula, $B(\oalga,\oalgb) = \one_{\Z'}$ and $A(\algc,\one_\algA) = \one_\algA$ for some element $\algc$ of an arbitrary algebra $\algA$. Then $B(\algc,\one_\algA)=\one_\algA$ . 
\end{lemma}
\begin{proof} Since $A(\algc,\one_\algA) = \one_\algA$, we have $\neg(\algc \land \one_\algA) = \neg \algc = \one_\algA$, that is, $\algc = \zero_\algA$. Application of Corollary \ref{BoolB} completes the proof.
\end{proof}


\begin{cor} \label{cor_d4} Suppose $B(p,q)$ is a formula, $B(\oalga,\oalgb) = \one_{\Z'}$ and $A(\algc,\algd) = \one_\algA$ where $\algc, \algd$ are some elements of an arbitrary s.i. algebra $\algA$. Then 
\begin{itemize}
\item[(a)] If $\algd \lor \neg \algd = \one_\algA$ then $B(\algc,\algd)=\one_\algA$;
\item[(b)] If $\algd = \neg \algc$ then $B(\algc,\algd)=\one_\algA$;
\item[(c)] If $\neg \neg \algd = \neg \algc$ then $B(\algc,\algd)=\one_\algA$.
\end{itemize}
\end{cor}
\begin{proof} (a) Indeed, since algebra $\algA$ is s.i., then $\algd \lor \neg \algd = \one _\algA$ yields $\algd = \one_\algA$ or $\algd = \zero_\algA$. Applications of lemmas \ref{lcase_1} and \ref{lcase_0} concludes the proof.

(b) Since $A(\algc,\algd) = \one_{\algA}$ we have
\begin{equation}
\begin{split}
& \one_\algA = ((\neg \neg \algc \impl \algc) \impl (\algd \lor \neg \algd)) = \\
& ((\neg \neg \algc \impl \algc) \impl (\neg \algc \lor \neg \neg \algc)) = ((\neg \neg \algc \impl \algc) \impl (\algc \lor \neg \algc)).   \label{caseb_1}
\end{split}
\end{equation}
From $A(\algc,\algd) = \one_{\algA}$ it also follows that
\begin{equation}
((\neg \neg \algc \impl \algc) \impl (\algc \lor \neg \algc)) \impl (\algd \lor \neg \algd) =\one_{\algA}. \label{caseb_2}
\end{equation} 
From \eqref{caseb_1} and \eqref{caseb_2} it trivially follows that $\algd \lor \neg \algd =\one_\algA$ and application of (a) completes the proof of case (b). 

(c) Immediately from $A(\algc,\algd) = \one_\algA$ it follows that $\neg \neg \algd \impl \algd = \one_\algA$, that is, $\neg \neg \algd = \algd$. Thus, $ \neg \neg \algd = \neg \algc$ yields $\algd = \neg \algc$ and we can apply (b).
\end{proof}

\subsection{The proof of the theorem}
\begin{proof}
In order to prove the theorem we will demonstrate that formula $A(p,q)$ and the valuation $\nu$ such that $\nu(p) = \oalga$ and $\nu(q) = \oalgb$ define algebra $\Z'$. It is clear that elements $\oalga, \oalgb$ generate algebra $\Z'$ and that $A(\oalga,\oalgb) = \one$, thus, the conditions (1) and (2) of the Definition \ref{fpdef} are satisfied. So, all what is left to prove is that for any formula $B(p,q)$ such that
\begin{equation}
B(\oalga, \oalgb) = \one \label{assumB_0}
\end{equation}
we have 
\begin{equation}
\vdash_{\classV(\Z^*)} A \impl B. \label{Vcons}
\end{equation}
In order to prove \eqref{Vcons} it is enough to show that for any s.i. homomorphic image $\algA$ of algebra $\Z^*$ and any two elements $\algc,\algd \in \algA$ if 
\begin{equation}
A(\algc, \algd) = \one _\algA, \label{assumpA}
\end{equation}
then
\begin{equation}
B(\algc, \algd) = \one_\algA. \label{assumpB}
\end{equation}
We will consider the following cases:

\begin{center}
\begin{tabular}{l|l|l}
\hline  & $\algc$ & $\algd$  \\
\hline (1) & $\algc \in Dn(\algA)$ & any  \\
\hline (2) & $\algc \in Rg(\algA)$ & any  \\
\hline (3) & $\algc \notin Rg(\algA)$ and $\algc \notin Dn(\algA)$ & any  \\
\end{tabular}
\end{center}

\subsubsection{Case (1).} If $A(\algc,\algd) = \one_\algA$ then $\neg (\alga \land \algc) =\one_\algA$, hence $\algc \land \algd = \zero_\algA$. Since $\algc \in Dn(\algA)$ we have $\algd = \zero$ and application of Lemma \ref{lcase_0} completes the proof.

\subsubsection{Case (2).} If $A(\algc,\algd) = \one_\algA$ then $(\neg \neg \algc \impl \algc)\impl (\algd \lor \neg \algd) = \one_\algA$. Since $\algc \in Rg(\algA)$, that is, $(\neg \neg \algc \impl \algc) = \one_\algA$, we can conclude that $\algd \lor \neg \algd = \one_\algA$ and apply Corollary \ref{cor_d4}.

\subsubsection{Case (3).} Let $\algc$ be neither regular, nor dense. Algebra $\algA$ is a s.i. homomorphic image of algebra $\Z^*$ and let $\nabla$ be a kernel of this homomorphism. Let us consider two cases:
\begin{itemize}
\item[(a)] $\lbr\one, \one \rbr \in \nabla$;
\item[(b)] $\lbr\one, \one \rbr \notin \nabla$.
\end{itemize}

(a) Let us recall that $\algA$ is a s.i. algebra, therefore $\lbr \zero, \one \rbr \lor \lbr \one, \zero \rbr = \lbr \one, \one \rbr \in \nabla$ yields $\lbr \zero,\one \rbr  \in \nabla$ or $\lbr \one,\zero \rbr  \in \nabla$.  

If $\lbr \zero,\one \rbr  \in \nabla$ then $\algA \cong \Z^*/\nabla = \Z'/\nabla$ is a two-element Boolean algebra and we can apply Lemma \ref{BoolB} (because $\algc \land \algd = \zero_\algA$ and, therefore,  $\algc = \zero_\algA$ or $\algd = \zero_\algA$).

If $\lbr \one,\zero \rbr  \in \nabla$ then $\algA \cong \Z^*/\nabla = \Z'/\nabla$ is a single-generated algebra. There is the only element of single-generated algebra which is not dense and regular, namely, its generator $\algg$. Let us observe that, since $A(\algg, \algd) = \one_\algA$, we have $\algg \land \algd = \zero_\algA$ and there are just two possibilities for $\algd$: either $\algd = \zero$, or $\algd = \algr_2$ (see Fig. 1). In the first case we can apply Lemma \ref{lcase_0}. In the second case we can apply Corollary \ref{cor_d4}(b), because $\algr_2 = \neg \algg$.  

(b) Since element $\algc \in \Z \times \Z_2$  is neither dense , nor regular, there are exactly three sub-cases:
\begin{itemize}
\item[i.] $\algc = \lbr\algf, \zero \rbr$, where $\algf \in Dn(\Z)$;
\item[ii.] $\algc = \lbr \algg, \one \rbr$;
\item[iii.] $\algc = \lbr \algg, \zero\rbr$.
\end{itemize}

i. If $\algc = \lbr\algf, \zero \rbr$, then 
\begin{equation}
\begin{split}
& (\neg \neg \algc \impl \algc) \impl (\algc \lor \neg \algc ) =  (\neg \neg \lbr\algf, \zero \rbr \impl \lbr\algf, \zero \rbr) \impl (\lbr\algf, \zero \rbr \lor \neg \lbr\algf, \zero \rbr ) = \\
& (\lbr \one,\zero\rbr \impl \lbr\algf, \zero \rbr) \impl (\lbr\algf, \zero \rbr \lor  \lbr\zero, \one \rbr) =  \lbr \algf , \one \rbr \impl \lbr \algf , \one \rbr = \one.
\end{split} \label{case_i}
\end{equation}
By assumption, $A(\algc,\algd) = \one$, hence, $((\neg \neg \algc \impl \algc) \impl (\algc \lor \neg \algc)) \impl (\algd \lor \neg \algd) = \one$. Therefore from \eqref{case_i} we have $\algd \lor \neg \algd = \one$ and we can apply Corollary \ref{cor_d4}(a).

ii. Let $\algc = \lbr \algg, \one \rbr$. From $\algc \land \algd = \zero$ it follows (see Fig. 1) that in this case $\algd = \zero$, or $\algd = \lbr\algr_2 ,\zero \rbr$. In the first case we can apply Lemma \ref{lcase_0}. In the second case we can apply Corollary \ref{cor_d4}(b), because $\lbr \algr_2,\zero \rbr = \neg \lbr \algg, \one\rbr = \algc$.

iii. Let $\algc = \lbr \algg, \zero \rbr = \oalga$. From $\algc \land \algd = \zero$ it follows that there are just four possibilities for $\algd$ (see Fig.1): $\algd \in \left\{ \zero, \lbr \zero,\one\rbr, \lbr \algr_2,\zero\rbr, \lbr \algr_2,\one\rbr\right\}$.
If $\algd = \zero$ we can apply Lemma \ref{lcase_0}. If $\algd = \lbr \zero,\one \rbr = \oalgd$, the statement trivially follows from the assumption \eqref{assumB_0}.

Let us observe the following (see Fig. 1):
\begin{equation}
(\neg \neg c \impl \algc) \impl (\algc \lor \neg \algc) = (\neg \neg \oalga \impl \oalga) \impl (\oalga \lor \neg \oalga) = \oalga^7. 
\end{equation}

Since $A(\algc,\algd) = \one$ we have 
\begin{equation}
(\neg \neg c \impl \algc) \impl (\algc \lor \neg \algc) \impl (\algd \lor \neg \algd) = \one,
\end{equation}
hence,
\begin{equation}
(\neg \neg c \impl \algc) \impl (\algc \lor \neg \algc) \leq (\algd \lor \neg \algd). 
\end{equation}
But
\begin{equation}
\begin{split}
& \algd \lor \neg \algd = \lbr \algr_2, \zero \rbr \lor \neg \lbr \algr_2, \zero \rbr = \lbr \algr_2, \zero \rbr \lor \lbr \algr_1, \one \rbr = \lbr \algr_2 \lor \algr_1 , \one \rbr; \\
& \algd \lor \neg \algd = \lbr \algr_2, \one \rbr \lor \neg \lbr \algr_2, \one \rbr = \lbr \algr_2, \one \rbr \lor \lbr \algr_1, \zero \rbr  = \lbr \algr_2 \lor \algr_1 , \one \rbr. \\
\end{split}
\end{equation}
The observation that $\lbr \algr_2 \lor \algr_1 , \one \rbr < \oalga^7$ completes the proof.
\end{proof}


\section{Finite presentability and concatenations}

In this section using concatenations of finitely presentable algebras we construct an infinite set of infinite finitely presentable algebras. 

\begin{theorem} \label{concatfp} Let $\classV$ be a variety of Heyting algebras and $\algA = \algA' + \Z_2 \in \classV$ and $\algB = \Z_2 + \algB' \in \classV$. Suppose algebras $\algA$ and $\algB$ are finitely presentable over $\classV$ and $\algA' + \algB' \in \classV$. Then algebra $\algA' + \algB'$ is finitely presentable over $\classV$.
\end{theorem}
\begin{proof}
Let pairs $\lbr A(p_1,\dots,p_n); \nu \rbr$ and $\lbr B(q_1,\dots,q_m); \mu \rbr$ are defining respectively algebras $\algA$ and $\algB$ and $\left\{ p_1,\dots,p_n\right\} \cap \left\{ q_1,\dots,q_m\right\} = \emptyset$. Assume that $A'(p_1,\dots,p_n)$ and $B'(q_1,\dots,q_m)$ are such formulas that $\nu(A') = \alga$ is a coatom of $\algA$ and $\mu(B') = \algb$ is an atom of $\algB$. Let us note that $\nu(A') = \mu(B')$. Then the pair $\lbr C; \phi \rbr$, where
\begin{equation} 
\begin{split}
& C(p_1,\dots,p_n,q_1,.\dots,q_m) = \\
& A(p_1,\dots,p_n) \land B(q_1,\dots,q_m) \land (A'(p_1,\dots,p_n) \eqv B'(q_1,\dots,q_m)) \\
& \phi(p_i) = \nu(p_i); i=1,\dots,n \text{ and } \phi(q_j) = \mu(q_j); j=1,\dots,m, 
\end{split} \label{def_form} 
\end{equation}
defines algebra $\algA'+\algB'$ over $\classV$.

Assume $\nu(p_1) = \alga_i; i=1,\dots,n$ and $\mu(q_j) = \algb_j; j=1,\dots,m$. It is easy to see that elements $G = \left\{ \alga_1,\dots,\alga_n,\algb_1,\dots,\algb_m \right\}$ generate algebra $\algA' + \algB'$. From the Proposition \ref{homext} it follows that in order to prove our claim it is enough to demonstrate that for any algebra $\algC \in \classV$ any mapping $\psi: G \mapsto \algC$ such that
\begin{equation}
	C(\psi(\alga_1),\dots,\psi(\alga_n),\psi(\algb_1),\dots,\psi(\algb_m)) = \one_\algC \label{psi}
\end{equation}
can be extended to a homomorphism $\overline{\psi}: \algA' + \algB' \impl \algC$.

Let us consider the following reducts of $\phi$: 
\begin{equation} 
\phi_1: p_i \mapsto \alga_i; i=1,\dots,n \text{ and } \phi_2: p_i \mapsto \algb_i; i=1,\dots,m. \label{restr}
\end{equation}
Let us recall now that algebras $\algA$ and $\algB$ are finitely presentable over $\classV$. Hence, mappings $\phi_1$ and $\phi_2$ can be extended to homomorphisms $\overline{\phi_1}: \algA \impl \algC$ and $\overline{\phi_2}: \algB \impl \algC$. From \eqref{def_form},\eqref{psi} and \eqref{restr} it follows that 
\begin{equation}
\overline{\phi_1}(\alga) = \overline{\phi_2}(\algb).
\end{equation}
Moreover,
\begin{equation}
\overline{\phi_1}(\alga') \leq \overline{\phi_1}(\alga) \text{ for all } \alga' \in \algA' \text{ and } \overline{\phi_2}(\algb) \leq \overline{\phi_2}(\algb') \text{ for all } \algb' \in \algB'.
\end{equation}
Thus we can construct a homomorphism $\overline{\psi}$ in the following way:
\begin{equation}
\overline{\psi}(\algc) = \begin{cases}
\overline{\phi_1}(\algc) \text{, when } \algc \in \algA';\\
\overline{\phi_2}(\algc) \text{, when } \algc \in \algB'.\\
\end{cases}
\end{equation}
$\overline{\psi}$ is a homomorphism because $\overline{\phi_1}$ and $\overline{\phi_2}$ are homomorphisms and for any $\alga' \in \algA'$ and $\algb' \in \algB'$
\begin{equation}
\begin{split}
& \alga' \land \algb' = \alga'; \\
& \alga' \lor \algb' = \algb'; \\
& \alga' \impl \algb' = \one; \\
& \algb' \impl \alga' = \alga'. \\ 
\end{split}
\end{equation} 
\end{proof}

\begin{cor} \label{con_fin} Let algebra $\algA = \algA' + \Z_2 \in \classV$ be finitely presentable over $\classV$. If $\algB$ is a finite algebra and $\algA' +\algB \in \classV$ then algebra $\algA' +\algB$ is finitely presentable over $\classV$.  
\end{cor}
\begin{proof} If $\algA' + \algB \in \classV$ then $\Z_2 + \algB \in \classV$. Since algebra $\Z_2 + \algB$ is finite, it is finitely presentable and we can apply the theorem.  
\end{proof}

\begin{cor} Variety $\classV = \classV(\Z \times \Z_2 + \Z)$ contains infinitely many infinite finitely presentable s.i. algebras.
\end{cor}
\begin{proof} From Theorem \ref{mainth} it follows that algebra $\Z'$ is finitely presentable over $\classV$. On the other hand, for all $n =1,2,\dots$ we have $\Z_2 + \Z_{2n+1} \in \classV$. By virtue of Corollary \ref{con_fin} all algebras $\Z' + \Z_{2n+1}; \ n=1,2,\dots$ are finitely presentable over $\classV$. 
\end{proof}

\section{Axiomatization by characteristic formulas}

It is well-known that not every variety (or every intermediate logic for this matter) can be axiomatized by Jankov formulas. In fact, there is a continuum of varieties that cannot be axiomatized by Jankov formulas \cite{Tomaszewski_PhD}[Corollary p.128]. In this section we will show that there is a continuum of varieties that cannot be axiomatized by Jankov formulas but, nevertheless, can be axiomatized by characteristic formulas. In order to do so we will construct an infinite independent set of characteristic formulas and then we will demonstrate that any subset of this set defines a variety that cannot by axiomatized by Jankov formulas.

First, let us observe the following simple criterion (the proof in terms of frames can be found, for instance, in \cite{Bezhanishvili_N_PhD}[Corollary 3.4.14(1)]).

\begin{prop} \label{Janax} A variety $\classV$ can be axiomatized by Jankov formulas if and only if  for every algebra $\algA \notin \classV$ there is such a finite algebra $\algB \in \CSub\CHom\algA$ that $\algB \notin \classV$.
\end{prop}

Variety $\classV$ is called \cite{MaltsevBook} \textit{locally finite} if for every $n$ there is a number $m$ such that every $n$-generated $\classV$-algebra contains less than $m$ elements.

\begin{cor}\cite{Bezhanishvili_N_PhD,Tomaszewski_PhD}. Every locally finite variety of Heyting algebras can be axiomatized by Jankov formulas. 
\end{cor}
\begin{proof} Let $\classV$ be a locally finite variety. It suffices to demonstrate that every $n$-generated algebra $\algA \notin \classV$ can be separated from $\classV$ by some Jankov formula. For finite algebras the statemnt is trivial, so we can assume that $\algA$ is an infinite algebra. Let $m$ be an upper bound of powers of $n$-generated algebras of $\classV$. By Kuznetsov Theorem \cite{Kuznetsov_1973} algebra $\algA$ contains chain subalgebras of any finite length. Thus, it contains a finite subalgebra of power greater than $m$. Hence, this subalgebra is not in $\classV$ and we can apply Proposition \ref{Janax}.  
\end{proof}

\begin{cor} If \(\classV\) is a variety and $\algA \in FPSI(\classV)$ is an infinite algebra, then formula $\chi(\algA)$ defines a variety that cannot be axiomatized by Jankov formulas. 
\end{cor}
\begin{proof} Let $\classV' = \left\{ \algB; \ \algB \vDash \chi(\algA) \right\}$ be a variety defined by formula $\chi(\algA)$. Due to Proposition \ref{Janax} it suffices to show that all finite members of $\CSub\CHom\algA$ belong to $\classV'$. For contradiction: assume that $\algB \in \CSub\CHom\algA$ is a finite algebra and $\algB \notin \classV'$. Then $\algB \nvDash \chi(\algA)$. Since $\algB \in \CSub\CHom\algA \subseteq \classV$, we can apply Theorem \ref{charf}(a) and conclude that $\algA \in \CSub\CHom\algB$ which is impossible because $\algA$ is an infinite algebra while $\algB$ is finite algebra and, therefore, all algebras from $\CSub\CHom\algB$ are finite.
\end{proof}

Moreover, in the similar way one can prove the following.

\begin{cor} \label{notJ} Let $\classV$ be a variety and $\classK$ be a set of infinite algebras from $FPSI(\classV)$. Then the set of formulas $\left\{ \chi(\algA); \ \algA \in \classK \right\}$ defines a variety that cannot be axiomatized by Jankov formulas. 
\end{cor}
\begin{proof} Let $\classV' = \left\{ \algB; \ \algB \vdash \chi(\algA), \ \algA \in \classK \right\}$ be a variety defined by formulas $\left\{ \chi(\algA); \ \algA \in \classK \right\}$. Due to Theorem \ref{Janax} it suffices to show that for some algebra $\algA \in \classK$ all finite members of $\CSub\CHom\algA$ belong to $\classV'$. For contradiction: assume that $\algB \in \CSub\CHom\algA$ is a finite algebra and $\algB \notin \classV'$. Then $\algB \nvDash \chi(\algC)$ for some $\algC \in \classK$. Since $\algB \in \CSub\CHom\algA \subseteq \classV$, we can apply Theorem \ref{charf}(a) and conclude that $\algC \in \CSub\CHom\algB$ which is impossible because $\algC$ is an infinite algebra while $\algB$ is finite algebra and, therefore, all algebras from $\CSub\CHom\algB$ are finite.
\end{proof}

\begin{remark} It is important to note that in the Corollary \ref{notJ} all algebras from $\classK$ are finitely presentable over the same variety $\classV$.
\end{remark}

\subsection{Varieties not axiomatizable by Jankov formulas}

\begin{theorem} \label{cont} There is a continuum of varieties that are axiomatized by characteristic formulas, but cannot be axiomatized by Jankov formulas.
\end{theorem}
\begin{proof} In order to prove the theorem we will demonstrate that there is such a variety $\classV$ that $FPSI(\classV)$ contains an infinite antichain of its infinite members. Indeed, if $\classK \subseteq FPSI(\classV)$ is an antichain of infinite algebras, then, by virtue of Corollary \ref{indep1}, the set of formulas $CH = \left\{ \chi(\algA); \ \algA \in \classK \right\}$ is independent. Thus, all the varieties defined by distinct subsets $CH' \subseteq CH$ are pairwise different. And, due to Corollary \ref{notJ}, no variety defined by any formulas from $CH$ can be axiomatized by Jankov formula.  

Let 
\begin{equation}
\algA = \Z \times \Z_2 + \prod_{n=3}^\infty( \Z_{2n} + \Z_2 + \Z_2). 
\end{equation}

Let us consider algebras
\begin{equation}
\algA_n = \Z \times \Z_2 + \Z_{2n} + \Z_2 + \Z_2 \ ; \ n=1,2,\dots.
\end{equation} 
We need to demonstrate
\begin{itemize}
\item[(a)] For every $k$ algebra $\algA_k$ is finitely presentable over $\classV(\algA)$;
\item[(b)] For any $m \neq k$ algebra $\algA_m \notin \CSub\CHom\algA_k$, that is, the set $\left\{ \algA_n; \ n=1,2,\dots \right\}$ forms an antichain.
\end{itemize}

(a) First, let us observe that for any $k$ algebra $\Z_{2k} + \Z_2 + \Z_2$ is a homomorphic image of the direct product 
\begin{equation*}
\algB = \prod_{n=3}^\infty( \Z_{2n} + \Z_2 + \Z_2).
\end{equation*}
Thus, for each $k$ there is such a filter $\nabla_k \subseteq \algB$ that $\algB/\nabla_k \cong \algA_k$. By virtue of Proposition \ref{concatf}, we have
\[
A_k = \Z \times \Z_2 + (\Z_{2k} + \Z_2 +\Z_2) \cong \Z \times \Z_2 + (\algB/\nabla_k) \cong (\Z \times \Z_2 + \algB)/\nabla_k) = \algA/\nabla_k. 
\]
So, $\algA_k \in \CHom\algA$, hence, $\algA_k \in \classV(\algA) = \classV$. Let us also observe that algebras $\Z \times \Z_2 + \Z_2$ and $\Z_2 + \Z_{2k} + \Z_2 +\Z_2$ are subalgebras of algebra $\algA_k$. Therefore $\Z \times \Z_2 + \Z_2 ,\Z_2 + \Z_{2k} + \Z_2 +\Z_2 \in \classV$. From the Theorem \ref{mainth} it follows that algebra $\Z \times \Z_2 + \Z_2$ is finitely presentable over $\classV$. On the other hand, algebra $\Z_2 + \Z_{2k} + \Z_2 +\Z_2$ is finite and is finitely presentable over $\classV$ too. Now we can apply Theorem \ref{concatfp} and conclude that algebra $\algA_k$ is finitely presentable over $\classV$.

(b) Let $k \neq m$. We need to demonstrate that $\algA_k \notin \CSub\CHom\algA_m$. Let $\nabla \subseteq \algA_m$ be a filter. Let us consider two cases: 
\begin{enumerate}
\item \(\nabla\) contains the top element of algebra $\Z \times \Z_2$; 
\item \(\nabla\) does not contain the top element of algebra $\Z \times \Z_2$.
\end{enumerate}

\textbf{Case 1.}  If \(\nabla\) contains the top element of algebra $\Z \times \Z_2$, then $\algA_k/\nabla \cong (\Z \times \Z_2)/\nabla'$, where $\nabla' = \nabla \cap (\Z \times \Z_2)$. It is not hard to see that algebra $\Z \times \Z_2$ has just 2 infinite homomorphic images, namely, itself and the algebra $\Z$. Let us observe that algebra \(\Z \times \Z_2\) is not a proper subalgebra of itself or of algebra \(\Z\). Hence, algebra \(\algA_k\) is not embeddable in any homomorphic image of \(\algA_m\) as long as its kernel contains the top element of algebra $\Z \times \Z_2$.

\textbf{Case 2.} The important point to note here is that in algebra $\algA_k$ the top element of algebra $\Z \times \Z_2$ is at the same time the bottom element of algebra $\Z_{2k}$. Thus, by virtue of Proposition \ref{concatf}, all considerations can be reduced to algebras $\Z_{2k} + \Z_2 +\Z_2$ and $\Z_{2m} + \Z_2 +\Z_2$. But it is well known (see, for instance \cite{Gerchiu_Kuznetsov,Wronski_Cadr_1974}) that if $k \neq m$, then $\Z_{2k} + \Z_2 + \Z_2 \notin \CSub\CHom(\Z_{2m} +\Z_2 + \Z_2)$.      
\end{proof}


\subsection{Varieties not axiomatizable by characteristic formulas}

As we saw in the previous section there is a continuum of intermediate logics that cannot be axiomatizable by Jankov formulas, but can be axiomatized by characteristic formulas. Naturally the question arises whether any intermediate logic can be axiomatized by characteristic formulas. In this section we give a negative answer to this question.

Let us recall a notion of pre-true formula introduced by A. V. Kuznetsov \cite{Kuznetsov_Gerchiu} and used also by A. Wronski \cite{Wronski_1973}: a formula $A$ is called \textit{pre-true} in algebra $\algA$ if it is not valid in $\algA$, but is valid in all proper subalgebras and homomorphic images of $\algA$.

It is not hard to see that if $A$ is a Jankov formula of some finite algebra $\algA$, then $A$ is pre-true in $\algA$.

We will say that an algebra $\algA$ is \textit{self-embeddable} if $\algA$ is a proper subalgebra of itself, or it is embeddable in some proper own homomorphic image. For instance, finite algebras are not self-embeddable, Rieger-Nishimura ladder $\Z$ is not self embeddable, while algebra $\Z_2 + \Z$ is self-embeddable.

\begin{theorem} Assume $\classV$ is a variety and $\algA \in \classV$ is a not self-embeddable s.i. algebra finitely presentable over $\classV$. Then characteristic formula of algebra $\algA$ over $\classV$ is a pre-true formula of algebra $\algA$. 
\end{theorem}
\begin{proof} Let $C$ be a characteristic formula of algebra $\algA$ over $\classV$. Then, by definition, $\algA \nvDash C$. We need to prove that $\algA' \vDash C$ for any proper subalgebra or homomorphic image $\algA'$ of algebra $\algA$. Assume the contrary: $\algA' \nvDash C$. Then $\algA$ is embeddable in some homomorphic image of $\algA'$ and, therefore, $\algA$ is embeddable in a proper subalgebra or a proper homomorphic image of itself and, thus, $\algA$ is self-embeddable.  
\end{proof}

\begin{example} As we proved in Theorem \ref{mainth}, algebra $\Z'$ is finitely presentable over any variety $\classV(\Z \times \Z_2 + \algA)$. Thus, if $C$ is a characteristic formula of algebra $\Z'$ over variety $\classV(\Z')$, then $C$ is a pre-true formula of algebra $\algA$. Moreover, if $\algA_1$ and $\algA_2$ are two non-isomorphic algebras and $C_1$ and $C_2$ are characteristic formulas of algebra $\Z'$ over varieties $\classV(\Z \times \Z_2 + \algA_1)$ and $\classV(\Z \times \Z_2 + \algA_2)$, then both formulas $C_1$ and $C_2$ are pre-true in $\Z'$ even though formulas $C_1$ and $C_2$ may not be equivalent. 
\end{example}

One of the important questions regarding characteristic formulas is which varieties (or intermediate logics) can be axiomatized by characteristic formulas\footnote{For characteristic formulas of finite algebras this problem is studied in \cite{Bezhanishvili_N_PhD}.}. The following proposition gives some examples when a variety cannot be axiomatized by characteristic formulas of finite algebras. 

\begin{prop}\label{pretrue} Suppose $\classV$ is a variety defined by a formula $A$ which is a pre-true formula of some infinite algebra\footnote{The examples of such formulas can be found, for instance, in \cite{Kuznetsov_Gerchiu,Wronski_1973}} $\algA$. Then variety $\classV$ cannot be defined by Jankov formulas. 
\end{prop}
\begin{proof} First, let us observe that since formula $A$ is pre-true on $\algA$, all the proper subalgebras, homomorphic images of $\algA$ and their subalgebras are in $\classV$. Hence, neither Jankov formula can separate algebra $\algA$ from $\classV$: if $X$ is a Jankov formula of some finite s.i. algebra $\algB$ such that it is valid on all algebras from $\classV$ but refutable on $\algA$, by well-known properties of Jankov formulas, we have that algebra $\algB$ is embeddable in some homomorphic image of algebra $\algA$. Thus, since $\algB$ is finite while $\algA$ is infinite, formula $X$ is refutable on some proper subalgebra or proper homomorphic image of algebra $\algA$ and $\algB \in \classV$, and the latter contradicts $\algB \nvDash X$.
\end{proof}

\begin{example} Let $C$ be a characteristic formula of algebra $\Z'$ over $\classV(\Z')$. Then formula $C$ defines a variety that cannot by defined by Jankov formulas of finite algebras.
\end{example}

Now we will prove the main theorem of this section.

\begin{theorem} There exist the intermediate logics that are not axiomatizable by characteristic formulas.
\end{theorem}
\begin{proof} Let us consider the intermediate logic defined by the following axiom (the logic \textbf{KG } from \cite{Bezhanishvili_N_G_de_Jongh_2008}):
\begin{equation}
(p \impl q) \lor (q \impl r) \lor ((q \impl r) \impl r) \lor (r \impl (p \lor q)) \label{dense}
\end{equation}
And let \(\classV\) be a corresponding variety. From \cite{Kuznetsov_Gerchiu}[Lemma 4] it follows that any finitely generated s.i. algebra from \(\classV\) is a concatenation of finite number of some 1-generated algebras. Thus, every infinite finitely generated generated s.i. algebra from \(\classV\) is a concatenation of finite number of 1-generated algebras at least one of which is infinite, i.e. at least one of which is \(\Z\). From \cite{Citkin_Not_2012}[Theorem 2.14] it immediately follows that neither \(\classV\), nor any of its subvarieties contain infinite finitely  presentable s.i. algebras. Therefore, if we demonstrate that variety \(\classV\) contains subvarieties that are not axiomatizable by Jankov formula, we can complete the proof.  

Let us consider the Kuznetsov-Gerchiu algebra \(\mathsf{KG} = \Z + \Z_7 +\Z_2\) (diagram and frame of which are depicted at Fig. 2) 

\[
\ctdiagram{
\ctnohead
\ctinnermid
\ctel 0,15,15,0:{}
\ctel 15,0,45,30:{}
\ctel 15,0,45,30:{}
\ctel 0,15,45,60:{}
\ctel 0,15,15,30:{}
\ctel 30,15,0,45:{}
\ctel 45,30,10,65:{}
\ctel 0,45,20,65:{}
\ctel 45,60,40,65:{}
\ctel 20,90,5,105:{}
\ctel 20,90,35,105:{}
\ctel 5,105,35,135:{}
\ctel 35,105,50,120:{}
\ctel 35,105,20,120:{}
\ctel 50,120,35,135:{}
\ctel 35,135,35,165:{}
\ctv 0,15:{\bullet}
\ctv 15,0:{\bullet}
\ctv 30,15:{\bullet}
\ctv 45,30:{\bullet}
\ctv 15,30:{\bullet}
\ctv 30,45:{\bullet}
\ctv 45,60:{\bullet}
\ctv 0,45:{\bullet}
\ctv 15,60:{\bullet}
\ctv 20,80:{\overbrace{. \ . \ . \ . \ . \ . \ . \ . \ .}^{}}
\ctv 20,90:{\bullet}
\ctv 5,105:{\bullet}
\ctv 35,105:{\bullet}
\ctv 20,120:{\bullet}
\ctv 50,120:{\bullet}
\ctv 35,135:{\bullet}
\ctv 35,150:{\bullet}
\ctv 35,165:{\bullet}
\ctel 150,0,150,30:{}
\ctel 150,30,180,60:{}
\ctel 150,30,135,45:{}
\ctel 180,60,160,80:{}
\ctel 135,45,170,80:{}
\ctel 150,105,150,150:{}
\ctel 180,105,180,150:{}
\ctel 150,150,180,135:{}
\ctel 150,135,180,120:{}
\ctel 150,120,180,105:{}
\ctel 150,105,155,100:{}
\ctel 180,150,150,120:{}
\ctel 180,135,150,105:{}
\ctel 180,120,175,115:{}
\ctel 180,105,175,100:{}
\ctv 150,0:{\bullet}
\ctv 150,30:{\bullet}
\ctv 165,45:{\bullet}
\ctv 135,45:{\bullet}
\ctv 180,60:{\bullet}
\ctv 165,75:{\bullet}
\ctv 150,90:{.}
\ctv 180,90:{.}
\ctv 150,95:{.}
\ctv 180,95:{.}
\ctv 150,100:{.}
\ctv 180,100:{.}
\ctv 150,105:{\bullet}
\ctv 180,105:{\bullet}
\ctv 150,120:{\bullet}
\ctv 180,120:{\bullet}
\ctv 150,135:{\bullet}
\ctv 180,135:{\bullet}
\ctv 150,150:{\bullet}
\ctv 180,150:{\bullet}
}
\]
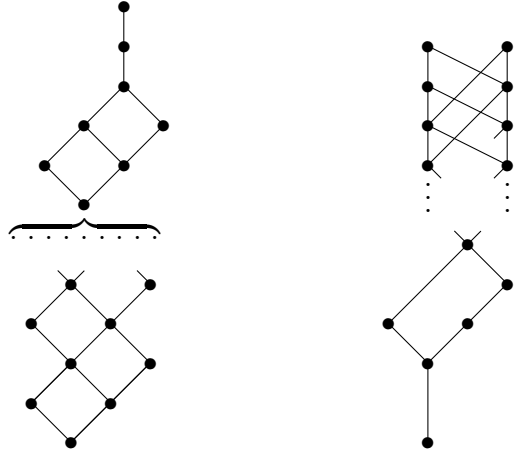
\begin{center} Figure 2. Kuznetsov-Ger{\v{c}}iu algebra \end{center}

This algebra has a pre-true formula, namely, the formula:
\begin{equation}
(((p \impl q) \impl q) \impl p) \impl (p \lor (p \impl q)) \impl (r \lor ((p \impl q) \lor ((p \impl q) \impl q)). \label{KGF}
\end{equation}
(cf. \cite{Kuznetsov_Gerchiu}[formula 2],\cite{Gerchiu_1972}[formula 4]). In \cite{Gerchiu_1972}[Lemma 2] it was proven that if the formula \eqref{dense} is valid in some algebra \(\algA\) then formula \eqref{KGF} is valid in \(\algA\) if and only if algebras \(\algA_1 = \Z_7 + \Z_2\) and \(\algA_2 = \Z_2 + \Z_7 + \Z_2 \) are not embeddable in \(\algA\). Let us observe that algebras \(\algA_1,\algA_2\) are not embeddable in any proper homomorphic image or any proper subalgebra of algebra \(\mathsf{KG}\). Thus formula \eqref{KGF} is a pre-true formula and, by virtue of Proposition \ref{pretrue}, the variety \(\classV\) has a subvariety that is not axiomatizable by Jankov formulas, namely, the subvariety defined by formula \eqref{KGF}. 
\end{proof}

\begin{remark} Using similar reasoning one can prove that all algebras \(\Z + \Z_{2n+1}+\Z_2\), where \(k > 2\), have pre-true formulas and obtain an infinite sequence of logics not axiomatizable by characteristic formulas. Let us note though that algebras  $\Z + \Z_{2n+1}+\Z_2\); \(k > 2$ (as opposed to algebras \(\Z \times Z_2 + \Z_{2n+1}+\Z_2\)) do not form an antichain: for instance, algebra $Z + Z_7 +Z_2$ is embeddable into algebra $Z + Z_9 +Z_2$. Thus, we cannot use algebras $\Z + \Z_{2n+1}+\Z_2$ in order to construct a continuum of logics not axiomatizable by characteristic formulas.  
\end{remark}


\section{Characteristic formulas of interior algebras}

In this section, using connections between varieties of Heyting and interior algebras, we will prove analogous results for varieties of interior algebras. 

\subsection{Basic definition}

Some facts regarding connections between varieties of Heyting and interior algebras that we will be using can be found in \cite{Maksimova_Rybakov,Chagrov_Zakh}.

We consider interior algebras in the signature $\land, \lor. \impl, \neg, \Box$. By $\Heyt$ we denote a variety of all Heyting algebras and by $\Topo$ we denote a variety of all interior algebras. Formulas without occurrences of $\Box$ we will call \textit{assertoric} as opposed to the \textit{modal} formulas in the extended signature. However, we will often omit ``modal'' if no confusion arises. If $A$ is an assertoric formula by $T(A)$ we denote the G\"odel-McKinsey-Tarski translation of $A$. If $A$ is a formula by $Var(A)$ we denote a set of all variables occurring in $A$. 

We will use the following notation and statements from (\cite{Maksimova_Rybakov}): if $\classM \subseteq \Topo$ is a variety then $\rho(\classM) \subseteq \Heyt$ is a variety defined by all assertoric formulas $A$ whenever $\vdash_\classM T(A)$. If $\algB$ is an interior algebra then $H(\algB)$ is the Heyting algebra of open elements of the algebra $\algB$ that we will call \textit{Heyting carcass of} $\algB$. Then, $\rho(\classM) = \left\{H(\algB): \algB \in \classM \right\}$ and $\rho$ is a homomorphism \cite{Maksimova_Rybakov} of the complete lattice of subvarieties of $\Topo$ onto complete lattice of subvarieties of $\Heyt$.

If $\algA$ is a Heyting algebra then \textit{a modal span of} $\algA$ (span for short) is the smallest relative to embeddings interior algebra $s(\algA)$, the Heyting carcass of which is isomorphic with $\algA$. The span of algebra $\algA$ can be constructed by taking the free Boolean extension $B(\algA)$ of $\algA$, and for each $\alga \in B(\algA)$ letting $\Box\alga = \bigwedge_{i=1}^n(\alga_i \impl \alga_i')$ , where $\alga = \bigwedge_{i=1}^n(\neg \alga_i \lor \alga_i')$ (see \cite[p. 191]{Blok_Dwinger_1975} or \cite[pp. 128-130]{Rasiowa_Sikorski}). Then $(B(\algA),\Box) $ is an interior algebra, indeed a $Grz$-algebra, where $Grz$ is the well known Grzegorczyk system. The Blok-Esakia Theorem establishes an isomorphism between lattices of varieties of Heyting and Grzegorczyk algebras. 

If $\algB$ is an interior algebra,	 then by $\algB^o$ we denote a subalgebra of $\algB$ generated by its open elements, that is, by elements of $H(\algB)$. In fact \cite[Lemma 3.4]{Maksimova_Rybakov}, $\algB^o$ is a modal span of $H(\algB)$.


\subsection{Finitely presentable interior algebras}

For interior algebras finite presentability can be defined in the following way.

\begin{definition}\label{fpdefi} (cf. \cite{MaltsevBook}) Let $\classM \subseteq \Topo$ be a variety of interior algebras, $A(\up)$ be a formula and $\nu$ be a valuation in algebra $\algA$. Then a pair $\lbr A, \nu \rbr$ \textit{defines algebra} $\algA$ over $\classI$ if
\begin{enumerate}
\item Elements $\alga_1 = \nu(p_1),\dots, \alga_n = \nu(p_n)$ generate algebra $\algA$;
\item $A(\ua) = \one$;
\item For any formula $B(\up)$ if $B(\ua) = \one$ then $\vdash_\classM(A(\up) \dimpl B(\up))$.
\end{enumerate} 
\end{definition}

The relation between finite presentability of Heyting algebra $\algA$ and interior algebra $s(\algA)$ can be expressed by following proposition.

\begin{prop} \cite{Citkin_Not_2012} \label{translf} Let $\classM$ be a variety of interior algebras and $\algA$ be a Heyting algebra. Algebra $s(\algA)$ is finitely presentable over $\classM$ if and only if algebra $\algA$ is finitely presentable over $\rho(\classM)$.
\end{prop}

The following theorem is a straight consequence of Proposition \ref{translf} and Theorem \ref{mainth}.

\begin{theorem} \label{mainthi} Let $\algA$ be any non-degenerate Heyting algebra and $\Z^* = \Z \times \Z_2 + \algA$. Then algebra $s(\Z')$, where $\Z' = \Z \times Z_2 + \Z_2$, is finitely presentable over variety $\classV(s(\Z^*))$.
\end{theorem}

\begin{example} Let $\algA= s(\Z \times \Z_2 + \Z_2))$ (see Fig.1). Obviously $\algA$ is subdirectly irreducible and, according to  Theorem \ref{mainthi}, algebra $\algA$ is finitely presentable over $\classV(\algA)$.
\end{example}

\begin{remark} In \cite[Section E]{Kracht_1990} M.Kracht shows the way how to construct an infinite set of infinite algebras that are splitting the variety corresponding to the logic of \textbf{S4}-frames of width 3.   
\end{remark}

Now we can define a characteristic formula for interior algebra similarly to how we did it for Heyting algebra.

\begin{definition} Assume $\classM$ is a variety of interior algebras and $\algA \in \classM$ is a s.i. algebra finitely presentable over $\classM$. Suppose $\lbr A(\up), \nu\rbr$ is a presentation of $\algA$ over $\classM$. If $B(\up)$ is such a formula that $\nu(B)$ is an opremum of $H(\algA)$, then the formula
\begin{equation}
\chi_{_\classM}(\algA) = A(\up) \dimpl B(\up)
\end{equation}
is a \textit{characteristic formula of algebra }$\algA$ \textit{over variety }$\classM$.
\end{definition} 

It is not hard to see that Theorem \ref{charf} holds true for interior algebra.

\begin{theorem} \label{charfi} Assume $\classM \subseteq \Topo$ , $\algA \in FPSI(\classM)$, $\algB \in \classM$ and $B$ is a formula. Then
\begin{itemize}
\item[(a)] if $\algB \nvDash \chi_{_\classM}(A)$, then $\algA \in \CSub\CHom\algB$;
\item[(b)] if $\algA \nvDash B$, then $B \vdash_\classM \chi_{_\classV}(A)$.
\end{itemize} 
\end{theorem}

Using Theorems \ref{charfi} and \ref{cont} one can prove the following theorem.

\begin{theorem} \label{conti} There is a continuum of varieties of interior algebras that are defined by characteristic formulas, but cannot be axiomatized by Jankov formulas.
\end{theorem}

\paragraph{Acknowledgments.} 
Many thanks to A. Muravitsky and G.Bexhanishvili for their pieces of advice and fruitful discussions.



\def\cprime{$'$}

\end{document}